%% file: main.tex
\documentclass[12pt, onecolumn,draftclsnofoot]{IEEEtran}%

\input{preamble}
\crefname{question}{question}{questions}

\usepackage{tikz}
\usepackage{pgfplots}
\usepackage{epigraph}
\usepackage{wrapfig}
\pgfplotsset{compat=1.5}

\providecommand{\U}[1]{\protect\rule{.1in}{.1in}}
\def\x{{\mathbf x}}

\theoremstyle{definition}

\begin{document}
\date{}
\title{The faces of Convolution: from the Fourier theory to algebraic signal processing}
\author{Feng~Ji, Wee~Peng~Tay,~\IEEEmembership{Senior Member,~IEEE} and Antonio Ortega,~\IEEEmembership{Fellow,~IEEE}%
\thanks{F.\ Ji and W.\ P.\ Tay are with the School of Electrical and Electronic Engineering, Nanyang Technological University, 639798, Singapore (e-mail: jifeng@ntu.edu.sg, wptay@ntu.edu.sg). A.\ Ortega is with the Department of Electrical and Computer Engineering, University of Southern California, Los Angeles, CA 90089-2564  (e-mail: aortega@usc.edu).}%
}
\maketitle

\section{Introduction and motivations}
\medskip

\setlength{\epigraphwidth}{0.93\textwidth}
\epigraph{You have three faces.
The first face, you show to the world.
The second face, you show to your close friends and your family.
The third face, you never show anyone. It is the truest reflection of who you are!}{A Japanese saying}

Convolution is a prevalent concept in both signal processing and machine learning. 
The early uses of the term ``convolution'' date back to the 1700s. In \emph{Recherches sur différents points importants du système du monde (1754)}, D'Alembert used the convolution integral to derive Taylor's theorem. The expression
\begin{align} \label{eq:ifg}
    \int f(u)g(x-u) \ud u
\end{align}
is used by S.\ Lacroix in \emph{Treatise on differences and series (1797-1800)}. Since then, the notion has appeared in numerous works, with different names. The term ``convolution'' was popularized in the 1950s-1960s to be the specific integral form in \cref{eq:ifg}. This is the \emph{first face} of convolution. 

In the 1930s-1940s, L.\ Pontryagin and other mathematicians (e.g., E.\ van Kampen and A.\ Weil) laid down the foundations of the theory of Pontryagin duality \cite{Mor77, Ram99}. The theory takes the point of view that the Fourier kernel $e^{-i2\pi\xi(\cdot)}$ defines a duality between  $\mathbb{R}$ and itself. The entire Fourier theory can be built from such a duality. The formalism preserves all familiar expressions in the classical theory, for example, convolution still takes the formal expression \cref{eq:ifg}. This vast generalization has numerous applications in harmonic analysis, signal processing, and number theory \cite{Ram99}. For example, the discrete Fourier transform (DFT) is the duality between the finite cyclic group $\mathbb{Z}/n$ and itself.

In the new age of data and artificial intelligence, convolution is associated with many concepts. For example, in computer vision, we have convolutional neural networks (CNNs) \cite{Lec98}, which involve pixel-wise multiplication with kernels of small size. In graph signal processing (GSP) and graph neural networks (GNN) \cite{Shu12, Def16}, we have the notion of graph convolution, which is closely related to graph message passing. For example, the above-mentioned DFT can also be interpreted as GSP on the directed cycle graph. With many seemingly different constructions bearing the same name, it is natural to ask whether there is any relationship between these concepts and if there is a way to place them within a common general framework. As an attempt, algebraic signal processing (ASP) \cite{Pus08} is proposed as a convolutional framework for such a purpose. 

However, if we compare ASP, an algebraic generalization of GSP, with the Pontryagin theory, convolutions no longer take any concrete form similar to that of (\ref{eq:ifg}), i.e., it shows up with a different face. ASP relies largely on the point of view that convolutions are linear transformations and they form a ring analogous to the polynomial ring. Therefore, it remains to elucidate the relations between the classical and more recent approaches. 

In this expository article, we provide a self-contained discussion of a selected list of theories, each embedded with the concept of ``convolution'': from the classical Fourier theory to the theory of ASP. We aim to achieve a few goals in this discussion: 
\begin{itemize}
    \item We carefully describe, with emphasis on mathematical rigor, the setup of each theory and explain how convolution is perceived. We illustrate with examples of how convolution is defined in each theory.
    \item By following the approximate historical timeline, we explain the relations among various notions of convolutions. We identify rigorous relationships between the different theories, where possible.
    \item We provide an opinion on whether there is a consistent approach to convolution under the current literature status. 
\end{itemize}
Convolution has many different ``faces''. Toward the end of this article, we describe what we see that might be common to what we know. The highlight is the observation that \emph{(informally) convolution determines the Fourier theory}. 

The rest of the paper is organized as follows. In \cref{sec:tcf}, we discuss the classical Fourier theory. Its generalization is the Pontryagin theory, which we provide a self-contained overview in \cref{sec:pd}. In \cref{sec:dtd}, we introduce GSP, which offers a different but equivalent interpretation of discrete Fourier transform in \cref{sec:pd}. Thus far, the notion of convolution has appeared in a few different forms. We compare them in \cref{sec:cg}. In \cref{sec:asp}, we describe the theory of ASP, which is a generalization of GSP. An overview of all the theories is provided in \cref{sec:au}. We give our point of view of a unification of different theories in terms of convolution. 

\section{The classical Fourier theory} \label{sec:tcf}
It is no exaggeration to say that the classical notion of convolution is embedded in the classical Fourier theory. Therefore, it is helpful to first properly describe the Fourier theory \cite{Rud87, Lax02}.

We denote the variables of (complexed valued) functions on the time domain and frequency domains, both can be identified by $\mathbb{R}$, by $t$ and $\omega$ respectively. Formally, the Fourier transform takes the integral expression
\begin{align}
    f \mapsto \widehat{f}, \text{ with } \widehat{f}(\omega) = \int_{-\infty}^{\infty} f(t)e^{-i2\pi t\omega}\ud t.
\end{align}

Apparently, the integral cannot be defined unconditionally for any function. We describe a few candidate domains of the Fourier transform as follows:
\begin{itemize}
    \item The space $C_c(\mathbb{R})$ of continuous functions with compact support. If $f \in C_c(\mathbb{R})$, then for any $\omega$, $f(t)e^{-i2\pi t\omega}$ is compactly supported and uniformly bounded (in $t$). Therefore, $\widehat{f}(\omega)$ is well-defined. Moreover, the Plancherel theorem holds:
    \begin{Theorem}
    \begin{align*}
        \int_{-\infty}^{\infty}|f(t)|^2\ud t = \int_{-\infty}^{\infty} |\widehat{f}(\omega)|^2\ud \omega.
        \end{align*}
    \end{Theorem}
    \item The space $L^1(\mathbb{R})$ of integrable functions. If $f \in L^1(\mathbb{R})$, then it is easy to verify that 
    \begin{align*}
    |\widehat{f}(\omega)| \leq \int_{-\infty}^{\infty} |f(t)e^{-i2\pi t\omega}|\ud t= \norm{f}_1
    \end{align*}
    as $|e^{-i2\pi t\omega}|=1$. Therefore, the Fourier transform defines a bounded linear map $L^1(\mathbb{R}) \to L^{\infty}(\mathbb{R})$. Moreover, the Riemman-Lebesgue lemma states that $\widehat{f}$ is a continuous function vanishing at $\infty$ if $f\in L^1(\mathbb{R})$. 
    \item The space $L^2(\mathbb{R})$ of square integrable functions. We notice that $C_c(\mathbb{R})$ is a dense subspace of $L^2(\mathbb{R})$. Therefore, for any $f\in L^2(\mathbb{R})$, we find a sequence $(f_i)_{i\geq \infty}$ of functions in $C_c(\mathbb{R})$ such that $f_i \to f$ as $i \to \infty$ in $L^2$-norm. By the Plancherel theorem, we observe that $\norm{f_i-f_j}_2 = \norm{\widehat{f_i}-\widehat{f_j}}_2$. As $(f_i)_{i\geq \infty}$ is a Cauchy sequence, so is $(\widehat{f_i})_{i\geq \infty}$. Then $\widehat{f}$ is defined as the $L^2$-limit of the Cauchy sequence $(\widehat{f_i})_{i\geq \infty}$. As a consequence, the Plancherel theorem holds for any $f \in L^2(\mathbb{R})$. 
    
    In this case, the Fourier transform $\widehat{\cdot}: L^2(\mathbb{R}) \to L^2(\mathbb{R})$ is:
\begin{enumerate}[(a)]
    \item Invertible: the inverse transform is defined by
    \begin{align*}
        f(t) = \int_{-\infty}^{\infty} \widehat{f}(\omega)e^{t2\pi \iu\omega}\ud \omega.
    \end{align*}
    \item Inner product preservation: for $f,g\in L^2(\mathbb{R})$, we have
    \begin{align*}
        \int_{-\infty}^{\infty} f(t)\widehat{g}(t) \ud t = \int_{-\infty}^{\infty} \widehat{f}(t)g(t) \ud t. 
    \end{align*}
\end{enumerate}
As a formal consequence, we conclude the following:
\begin{Theorem}
    The Fourier transform is a unitary operator on the Hilbert space $L^2(\mathbb{R})$. 
\end{Theorem}
\end{itemize}

More generally, the Fourier transform can be defined for tempered distributions, which is the dual of the Schwartz space. We shall not pursue this in the article and interested readers can find details in \cite{Lax02}.

Having described the fundamental framework, we can now introduce convolutions. It is the safest to consider $f,g \in C_c(\mathbb{R})$, their \emph{convolution} is defined by the following integral 
\begin{align}\label{eq:h=f}
    h(t) = f*g(t) = \int_{-\infty}^{\infty} f(z)g(t-z) \ud z.
\end{align}
The resulting function $h(\cdot)$ also belongs to $C_c(\mathbb{R})$. Intuitively, the convolution $h(\cdot)$ can be viewed as a weighted sum of $f(\cdot)$ with weight $g(\cdot)$ and vice versa. 

We may apply the same integral formula for $f$ and $g$ from other function spaces. For example, if $f,g\in L^1(\mathbb{R})$, then $h=f*g \in L^1(\mathbb{R})$ as well and $\norm{h}_1\leq \norm{f}_1\cdot \norm{g}_1$. On the other hand, if $f,g\in L^2(\mathbb{R})$, then $h=f*g$ does not necessarily belong to $L^2(\mathbb{R})$; instead, $f*g \in L^{\infty}(\mathbb{R})$. The convolution theorem relates convolution and Fourier transform. 
\begin{Theorem} \label{thm:ifi}
If $f*g \in L^1(\mathbb{R})$ or $L^2(\mathbb{R})$, then
    \begin{align*}
        \widehat{f*g} = \widehat{f}\cdot \widehat{g}. 
    \end{align*}
    Moreover, $L^1(\mathbb{R})$ is a commutative and associative algebra with convolution $*$ as the multiplication. 
\end{Theorem}

From this result, to generalize this classical notion of convolution to functions on a general domain $A$, there are two options: 
\begin{itemize}
    \item Formulate the integral as in (\ref{eq:h=f}). 
    \item Use the Fourier transform via the convolution theorem.
\end{itemize}
For the former approach, we need a translation operation ``$-$'' in $A$, while for the latter option, we need to make sense of the Fourier transforms, in particular, find an appropriate replacement of the integration kernel $e^{-t2\pi \iu\omega}$ in $A$. In both cases, we require a measure on $A$. All these can be accomplished for $A$ being a locally compact abelian group as we shall discuss in the next section (\cite{Mor77, Ram99}). 

\section{Pontryagin duality} \label{sec:pd}

Recall from the previous section that in order to introduce a Fourier theory and convolution for functions on a domain $A$, we need the following ingredients analogous to the classical theory:
\begin{enumerate}
    \item Algebraic condition: there is a translation on $A$.
    \item Analytic condition: there is a measure on $A$ to perform integrations.
    \item A family of ``Fourier kernels''.
\end{enumerate}

We go through these one-by-one. To fulfill the algebraic condition, we require $A$ to be an abelian group. More details regarding algebraic concepts introduced in the paper can be found in standard textbooks \cite{Ati94, Lan02, Hun03, Art11}. 

\begin{Definition}
$A$ is an abelian group if there is a binary operation $+: A\times A \to A$, called the addition, such that the following holds
\begin{itemize}
\item Associativity: $a+b+c = a+(b+c)$ for $a,b,c\in A$.
\item Commutativity: $a+b=b+a$ for $a,b\in A$.
    \item Existence of the identity: there is an element $0 \in A$ such that $0+a=a$ for any $a\in A$.
    \item Existence of inverses: for $a\in A$, there is $-a\in A$ such that $a+(-a)=0$. 
\end{itemize}
\end{Definition}

\begin{Example}
    \begin{enumerate}
        \item $\mathbb{Z},\mathbb{R}$ and $\mathbb{C}$ are abelian groups with the usual addition $+$.
        \item Let $S^1$ be the set of all complex numbers of (complex) norm $1$. Then it is an abelian group under the complex multiplication $\times$, with $1$ as the identity element. $S^1$ is called the circle group, as it can be visualized as the unit circle on the complex plane. 
        \item Let $\mathbb{Z}/n\mathbb{Z}$ be the finite set $0,\ldots, n-1$. Define $+$ on $\mathbb{Z}/n\mathbb{Z}$ by taking sum as integers modulo $n$It is straightforward to check that $\mathbb{Z}/n\mathbb{Z}$ is an abelian group with $0$ as the identity. It is the cyclic group of order $n$. 
    \end{enumerate}
\end{Example}

As a consequence of the setup, we can perform translation on $A$, which is nothing but addition with $-a$ for $a\in A$. When $A$ is $\mathbb{R}$, this agrees with the usual notion of translation. To proceed with the analytic condition on the existence of measure, we first impose a topological condition on $A$.

\begin{Definition}
An abelian group $A$ is locally compact if there is a topology $\mathcal{T}$ on $A$ such that the following holds
\begin{enumerate}
    \item The binary operation $+: A\times A \to A$ is continuous (we endue $A\times A$ with the product topology).
    \item Taking inverse $-: A \to A, a\mapsto -a$ is continuous.
    \item Each $a \in A$ has a compact neighborhood, i.e., there is an open set $U$ and compact set $K$ such that $a\in U \subset K$.  
\end{enumerate}
\end{Definition}

For example, $\mathbb{R}, \mathbb{C}$, and $S^1$ are locally compact given the usual Euclidean topology. While $\mathbb{Z}$ and $\mathbb{Z}/n\mathbb{Z}$ are locally compact given the discrete topology. An important property of a locally compact abelian group is the existence of Haar measure.

\begin{Theorem}
    If $A$ is a locally compact abelian group, there is a translation invariant Radon measure $\mu$, called the Haar measure, on $A$. Moreover, it is unique up to a scalar multiple.  
\end{Theorem}

Translation invariant in the theorem means that for any measurable set $U$ and $a\in A$, we have $\mu(U)=\mu(U+a)$. The Lebesgue measure on each of $\mathbb{R}, \mathbb{C}$ and $S^1$ is the Haar measure. On $\mathbb{Z}$ and $\mathbb{Z}/n\mathbb{Z}$, the discrete measure is the Haar measure. 

For the last ingredient, we need to understand the Fourier kernel $e^{-t2\pi \iu\omega}$ in the new context. Strictly speaking, it is a common misconception to review $\{e^{-(\cdot)2\pi \iu\omega},\omega\in \mathbb{R}\}$ as an eigenbasis for some linear operator on a Hilbert space such as $L^2(\mathbb{R})$. This is because it is not (square) integrable. Moreover, this set of functions is uncountable. Maybe a more appropriate interpretation is to view it as a character, which we introduce next. 

\begin{Definition}
    A function $\phi: A_1 \to A_2$ from an abelian group $A_1$ to $A_2$ is called a \emph{homomorphism} if $\phi(a+b)=\phi(a)+\phi(b)$ for $a,b\in A_1$, i.e., $\phi$ preserves the addition. A bijective homomorphism is called an isomorphism.
    
    A \emph{character} of a locally compact abelian group $A$ is a continuous homomorphism $\chi: A \to S^1$. 
    The set $\widehat{A}$ of all characters of $A$ is called the \emph{Pontryagin dual} of $A$. 
\end{Definition}

Two isomorphic groups are essentially the same in the algebraic sense, while continuous homomorphism also respects the topologies. 

Apparently, $\widehat{A}$ is also an abelian group with pointwise multiplication in $S^1$. It carries the compact-open topology if $A$ is locally compact. 

\begin{Example} \label{eg:ia=}
    \begin{enumerate}
        \item If $A = \mathbb{R}$, then for each $\omega$, we have a character $\chi_{\omega}: \mathbb{R} \to S^1$ by the formula $\chi_{\omega}(t) = e^{t2\pi \iu\omega}$. On the other hand, for a character $\chi: \mathbb{R} \to S^1$, let $\chi(1) = e^{2\pi \iu \omega} \in S^1$. As $\chi$ is a homomorphism, one can easily verify that $\chi(t) = e^{t2\pi \iu \omega}$ for $t\in \mathbb{Q}$. Moreover, since $\chi$ is continuous, the same formula holds for any $t \in \mathbb{R}$, i.e., $\chi = \chi_{\omega}$. In conclusion, $\widehat{\mathbb{R}}$ can be identified with $\mathbb{R}$.  
        \item If $A = \mathbb{Z}$, then one can use the same argument to show that $\widehat{\mathbb{Z}} = \mathbb{R}/\mathbb{Z} \cong S^1$ using the fact that $e^{2\pi \iu \omega} = e^{2\pi \iu(\omega+m)}$ for any $m\in \mathbb{Z}$. Conversely, it is also true that $\widehat{S^1} = \mathbb{Z}$. 
        \item \label{it:ia=} If $A = \mathbb{Z}/n\mathbb{Z}$, then a character $\chi: A \to S^1$ is determined uniquely by $\chi(1) = e^{2\pi \iu \omega}$ with $\omega$ is any of number of the form $i/n, 0\leq i\leq n-1$. Therefore, $\widehat{A}$ can be identified with $A$. 
    \end{enumerate}
\end{Example}

The first example suggests interpreting the kernel of the Fourier transform as a character of the group $\mathbb{R}$. Inspired by this observation, we are prompted to define the Fourier transform analogously as an integral.

\begin{Definition}
    For $f\in L^1(A)$, we define its Fourier transform $\widehat{f}: \widehat{A} \to \mathbb{C}$, a function on $\widehat{A}$, by the formula 
    \begin{align*}
        \widehat{f}(\chi) = \int_{A} f(a)\overline{\chi}(a) \ud a,
    \end{align*}
    where $\chi \in \widehat{A}$ and $\ud a$ is the Haar measure on $A$. 
\end{Definition}

To be rigorous, let $V(A)$ be the complex span of continuous functions of positive type on $A$. We do not formally define this technical term here. However intuitively, the positive condition resembles that of the positive semi-definiteness of matrices. We summarize the main results as follows (cf. \cite{Ram99}), many of which can find their counterparts in the classical Fourier theory. 

\begin{Theorem}
    \begin{enumerate}
        \item (Fourier inversion) There exists a Haar measure $\ud \chi$ on $\widehat{A}$ such that for all $f\in V^1(A) = V(A)\cap L^1(A)$,
        \begin{align*}
            f(a) = \int \widehat{f}(\chi)\chi(a) \ud \chi.
        \end{align*}
        Moreover, the Fourier transform $f\mapsto \widehat{f}$ identifies $V^1(A)$ with $V^1(\widehat{A})$. 
        \item (Pontryagin duality) $A$ is self-dual, i.e., the evaluation map $\alpha: A \to \widehat{\widehat{A}}, \alpha(a)(\chi) = \chi(a)$ is a continuous isomorphism.
        \item (The Plancherel theorem) The Fourier transform can be extended to a unitary isometry of Hilbert spaces from $L^2(A)$ onto $L^2(\widehat{A})$. 
    \end{enumerate}
\end{Theorem}

These results describe a complete picture of what we expect from a Fourier theory. In the following table, we summarize signal processing theories corresponding to different choices of $A$ and $\widehat{A}$:
\begin{center}
\begin{tabular}{ |c|c|c| } 
 \hline
 Theory & $A$ & $\widehat{A}$ \\ 
 \hline
 \hline
 Fourier transform & $\mathbb{R}$ & $\mathbb{R}$ \\ 
 \hline
 Fourier series  & $S^1$ & $\mathbb{Z}$ \\ 
 \hline
 Discrete-time Fourier transform (DTFT) & $\mathbb{Z}$ & $S^1$ \\ 
 \hline
 Discrete Fourier transform (DFT) & $\mathbb{Z}/n\mathbb{Z}$ & $\mathbb{Z}/n\mathbb{Z}$ \\ 
 \hline
\end{tabular}
\end{center}

It is worth mentioning that the framework has many other important applications such as in number theory. A detailed discussion is beyond the scope of the paper and interested readers may consult \cite{Tat65, Ram99}.

With all the hard work that has laid down the foundation, introducing convolution is just a matter of formality. For $f,g\in L^1(A)$, their \emph{convolution} is defined as 
\begin{align} \label{eq:fga}
    f*g(a) = \int g(a-b)f(b) \ud b. 
\end{align}
The convolution $f*g$ exists for almost all $a\in A$ and $f*g \in L^1(A)$ satisfying $\norm{f*g}_1 \leq \norm{f}_1\norm{g}_1$. If $f$ and $g$ are sufficiently regular, then the convolution theorem holds: $\widehat{f*g} = \widehat{f}\cdot \widehat{g}$.

\begin{Corollary} \label{coro:lia}
    $L^1(A)$ is a commutative and associative (Banach) algebra with convolution $*$ as the multiplication. 
\end{Corollary}

\section{DFT, the directed cycle graph and GSP} \label{sec:dtd}

As we have seen, a special case of the Pontryagin theory is DFT, when $A = \widehat{A} = \mathbb{Z}/n\mathbb{Z}$. In this case, $A$ is a finite set and hence all the function spaces on $A$ of interest can be identified $\mathbb{C}^n$. Hence, linear operators such as the Fourier transform and convolutions can be represented by matrices. For example, as we have seen in \cref{eg:ia=}, the characters of $A$ are $e^{2\pi \iu j/n}, 0\leq j\leq n-1$. Therefore, the Fourier transform can be represented by the matrix $\bU$ with the $i,j$-th entry $\bU_{i,j} = e^{i2\pi \iu j/n}$. 

For the cyclic group $A$, we may also visualize it combinatorically. It can be generated by the element $1$. The Cayley graph (\cite{Cay78}) $C_n$ associated with this generator is the directed cycle graph of size $n$, i.e., there is a directed edge from the node labeled by $i$ to the node labeled $i+1$ in $A$. The matrix of eigenvectors of the adjacency matrix $\bA_{C_n}$ of $C_n$ is nothing but $\bU$. Therefore, if we interpret a vector in $\mathbb{C}^n$ as functions, or a signal, on the nodes of $C_n$, then the DFT Fourier transform can be viewed as the linear transformation of the signal space by the matrix of eigenvectors of $\bA_{C_n}$. According to the convolution theorem, convolution can be defined by componentwise multiplication of the transformed signals.  

If we extract the essential ingredients of the above combinatoric approach, we need
\begin{itemize}
    \item a finite graph $G=(V,E)$ of size $|V|=n$
    \item an operator (matrix) $\bS$ associated with $G$ such that $\bS$ has an orthonormal (or more generally, unitary) eigenbasis.
\end{itemize}
The operator $\bS$ usually imitates a message passing on $G$, hence called a \emph{graph shift operator}. Typical examples of $\bS$ include the adjacency matrix of $G$, the Laplacian matrix of $G$, or their normalized version. The theory of graph signal processing (GSP) \cite{Shu12, Shu13, San13, San14, Gad14, Don16, Def16, Kip16, Egi17, Gra18, Ort18, Girault2018, Ji19} can be built upon such an $\bS$. As we deal with finite graphs, we do not need to worry about integrability. Therefore, it is usually sufficient to use tools from linear algebra.  

For the formal description of GSP, by the assumption, $\bS$ has a decomposition $\bS = \bU\bLambda \bU^{-1}$. A graph signal is a function $V \to \mathbb{C}$. As $V$ is finite, such a signal can be identified with a vector $\bx$ in $\mathbb{C}^n$. Each component of $\bx$ corresponds to a node $v\in V$. Therefore, the domain of $\bx$ is called the graph domain. Usually, a linear transformation of the space of graph signals is also called a \emph{filter}. Here we consider complex signals inconsistent with the classical Fourier theory, though most GSP literature and applications use real signals. 

The graph Fourier transform (GFT) of a graph signal $\bx$ is 
\begin{align*}
    \widehat{\bx} = \bU^{-1}\bx.
\end{align*}
Its inverse is given by the inverse graph Fourier transform (IGFT)
\begin{align*}
    \bx = \bU\widehat{\bx}.
\end{align*}
The codomain of GFT is called the frequency domain. Each component of the frequency domain corresponds to an eigenvector $\bu_i$ of $\bS$. The vector $\bu_i$ is also a graph signal and intuitively, it corresponds to a certain mode of signal fluctuation on $G$. 

To introduce convolution, it is most convenient to take the analogy of the equivalent formulation in the classical convolution theorem \cref{thm:ifi}. More precisely, for two graph signals $\bx,\by$, then their convolution $\bz = \bx*\by$ is determined by 
\begin{align} \label{eg:wzw}
    \widehat{\bz} = \widehat{\bx}\odot \widehat{\by},
\end{align}
where $\odot$ is the Hadamard product or component-wise multiplication. 

If $\bx$ is fixed, then taking convolution with it defines a linear transformation, also called a \emph{convolution filter}, on the signal space: $\bx*: \mathbb{C}^n \to \mathbb{C}^n, \by \mapsto \bx*\by$. The convolution theorem for GSP reads as follows.
\begin{Theorem}
    If $\bS$ does not have repeated eigenvalues, then for each convolution filter $\bx*$, there is a unique polynomial $P_{\bx}(\cdot)$ of degree at most $n-1$ such that the matrix form of $\bx*$ is $P_{\bx}(\bS)$. Therefore, the set of convolution filters agrees with the polynomial algebra of matrices in $\bS$.   

    Moreover, a filter $\bM$ is shift-invariant w.r.t.\ $\bS$, i.e., $\bM\circ \bS = \bS \circ \bM$, if and only if it is a convolution filter.  
\end{Theorem}

This result is useful as in theory, it gives yet another point-of-view of the notion of convolution in the graph context. In practice, to perform convolution, explicit eigendecomposition is not needed and one only has to compute polynomials in $\bS$. This also facilitates learning and estimation as a parameterized model. Many interesting filters belong to the convolution family including the band-pass filters. Moreover, they are extensively used in graph neural networks to be discussed in the next section. 

In this section, we consider $\mathbb{C}$ as the base field for the connection with the previous sections. In most applications, it is sufficient to use $\mathbb{R}$ as the base field, which will be assumed for subsequent sections unless otherwise stated.  

\section{CNN v.s.\ GNN} \label{sec:cg}
The term ``convolution'' is omnipresent in the era of artificial intelligence \cite{Lec98, Goo16, He16, Kri17, Def16}. In this section, we compare its usage in CNN for the field of computer vision and GNN for graph neural networks. 

For CNN, we shall only discuss the fundamental architecture eyeing to explain the meaning of convolution (based on \cite{Lec98}). We model an image by a compactly supported function on the (infinite) $2$D-lattice $A = \mathbb{Z}\oplus \mathbb{Z}$ as follows. If the size of the image is $w\times h$, we fit the lower left corner of the image with the origin and each pixel corresponds to a node on the integer lattice. (as illustrated in \figref{fig:cnn}). We construct a function $f$ on $A$ as follows: $f(a,b)$ is the pixel value at the $(a+1,b+1)$-th position of the image if $0\leq a\leq w-1, 0\leq b\leq h-1$; and otherwise, $f(a,b)=0$. 
\begin{figure}
    \centering
    \includegraphics[scale=0.6]{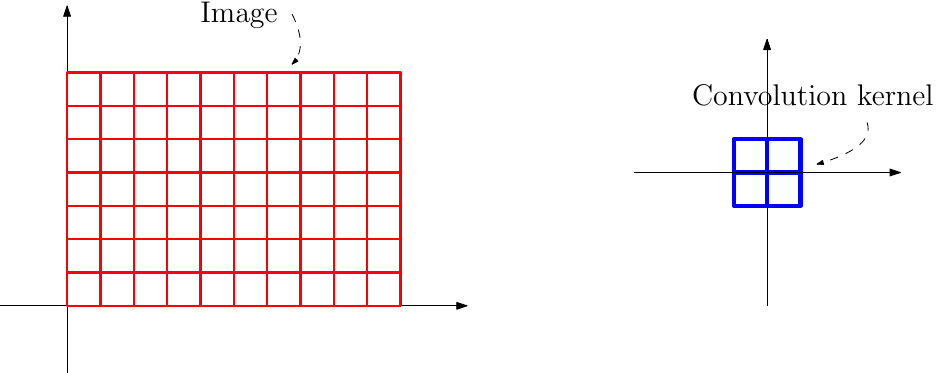}
    \caption{The setup for the basic CNN architecture.}
    \label{fig:cnn}
\end{figure}

On the other hand, for the proto-typical example of a $3\times 3$ convolution kernel, we can model it by any function $g$ on $A$ supported on the $9$ lattice points $\{-1,0,1\}\times \{-1,0,1\}$ (as illustrated in \figref{fig:cnn}). Notice that $A$ is a locally compact abelian group. Then the convolution of the image in CNN with stride $1$ is exactly $f*g$ given by the formula (\ref{eq:fga}) in \cref{sec:pd}. 

Encouraged by the success of convolution in computer vision, it is natural to anticipate similar architectures for graph-structured data. However, there is the major challenge that for a general graph $G$, unlike the 2$D$-lattice, each node may have completely different neighborhood structures, for example, different nodes usually have different degrees. The community quickly came to the consensus that convolution from GSP (\cref{sec:dtd}) should be the ``correct counterpart'' to design graph neural networks. As a natural consequence, the following GCN model \cite{Def16} is proposed, which serves as the prototype for most of the more sophisticated architectures \cite{Kip17, Ham17, Vel18, Cha19, Gul19}. 

For an undirected graph $G$, fix a graph shift operator $\bS$, such as the Laplacian or the normalized Laplacian of $G$ with self-loop. Suppose $\bX$ is a feature matrix such that the $i$-th row of $\bX$ is the feature vector associated with the $i$-th node. Set $\bW$ to be a learnable re-scaling matrix and $P$ to be a polynomial with learnable coefficients. Then the building block of the GCN model is the following graph convolution layer 
\begin{align*}
    \theta(\bX) = P(\bS)\bX\bW. 
\end{align*}

Rather than discussing various improvements of this fundamental model, which is not the aim of this paper, we want to compare ``convolution'' in CNN and GNN. For simplicity, we only consider information aggregation from the immediate neighbors of each node and there is no feature scaling. Then for CNN, we use a $3\times 3$ convolution filter, while for GNN, $P$ has degree $1$. An apparent difference is that the former has $9$ degrees of freedom, while the latter has only $2$ degrees of freedom. The implicit abelian group structure of a $2$D-lattice makes the neighborhood of each node identical. More generally, each node can be viewed as the center of the lattice. Therefore, the convolution for CNN exploits this property such that the kernel can be both local and sufficiently expressive. While for GNN, due to irregular neighborhood structures, the convolution compromises by assigning the same weight to the contribution of every neighborhood of all the nodes. Intuitively, the convolution for CNN is ``directional'' while that of GNN is ``radial''. 

\section{Algebraic signal processing} \label{sec:asp}
GSP has the simplicity that the entire theory can be established given a single operator $\bS$. The graph structure is only used once to construct such a matrix and the rest of GSP follows a formal procedure in (linear) algebra. The essential ingredients of the algebraic approach are reorganized to formulate the theory of algebraic signal processing (ASP) \cite{Pus08}, which founds applications outside the realm of graphs. 

We remark that strictly speaking, it is inaccurate to claim that ASP is derived from GSP. In the literature, ASP was proposed before the formal introduction of GSP. However, it is more natural to view ASP as a generalization of GSP as we discussed above. 

To present ASP, we first introduce a few algebraic concepts.
\begin{Definition}
A \emph{ring} $R$ is a set with two binary operations $+,\times$ with the following properties
\begin{enumerate}
    \item $R$ is an abelian group with $+$ with the identity $0$. 
    \item The multiplication $\times$ on $R$ isassociative. By convention, we may also denote $\times$ by $\cdot$ or omit it if no confusion arises.
    \item The multiplication $\times$ is distributive w.r.t.\ $+$, i.e., $r_1\cdot (r_2 + r_3) = r_1\cdot r_2 + r_1\cdot r_3$ and $(r_2 + r_3)\cdot r_1 = r_2\cdot r_1 + r_3\cdot r_1$.
\end{enumerate}   
$R$ is called commutative if $\times$ is commutative. It is called \emph{unital} if there is a multiplicative identity $1$ such that $1\cdot r = r\cdot 1 = r$ for any $r\in R$.
\end{Definition}

\begin{Example} \label{eg:auc}
    \begin{enumerate}
        \item A unital commutative ring $R$ is called a \emph{field} if $R\backslash \{0\}$ is also an abelian group under $\times$. For example, $\mathbb{R}$ and $\mathbb{C}$ are both fields, while the ring of integers $\mathbb{Z}$ is not a field as elements such as $2$ are not invertible in $\mathbb{Z}$. 
        \item \label{it:tso} The set of $n\times n$ matrices $M_n(\mathbb{R})$ is a unital ring under the usual matrix addition and multiplication. However, matrix multiplication is not commutative. 
It can be identified with the \emph{endomorphism ring} $\text{End}(\mathbb{R}^n)$ of linear transformation from $\mathbb{R}^n$ to itself, with composition as the multiplication. 
        \item Let $R$ be a commutative ring such as $\mathbb{Z},\mathbb{R}$ and $\mathbb{C}$. Then the set of polynomials $R[t]$ with coefficients in $R$ forms the (unital) \emph{polynomial ring}.
        \item \label{it:tsoi} The set of integrable functions $L^1(\mathbb{R})$ is a commutative ring with the usual function addition and convolution $*$ as the multiplication. However, it is not unital, while only has an approximate identity. 
    \end{enumerate}
\end{Example}

While the notion of ring generalizes that of field such as $\mathbb{R}, \mathbb{C}$, we have a corresponding generalization of the notion of vector space.

\begin{Definition}
    Given a ring $R$, an abelian group $M$ is a \emph{module over $R$} if there is a scalar multiplication $\cdot: R\times M \to M$ (the symbol ``$\cdot$'' is often omitted) such that the following holds for $r,r'\in R$ and $x,y\in M$
\begin{enumerate}
    \item $r\cdot (x+y) = r\cdot x + r\cdot y$.
    \item $(r+r')\cdot x = r\cdot x + r'\cdot x$.
    \item $(rr')\cdot x = r\cdot (r'\cdot x)$.
    \item $1\cdot x = x$.
\end{enumerate}
A subset $M'$ of $M$ is called a submodule over $M$ if $rx\in M'$ for any $x\in M'$. $M$ is irreducible if the only submodule of $M$ are the obvious ones $\{0\}$ and $M$.
\end{Definition}

\begin{Definition}
    If we have two modules $M_1, M_2$ over $R$, there \emph{direct sum} $M_1\oplus M_2$ consists of pairs $(x_1,x_2), x_1\in M_1, x_2\in M_2$. It is a also a module over $R$ with $r(x_1,x_2) = (rx_1,rx_2)$.
\end{Definition}

\begin{Example} \label{eg:tfs}
\begin{enumerate}
    \item \label{it:tvs} The vector space $\mathbb{R}^n$ is a module over $\mathbb{R}$. Moreover, it is also a module over the matrix ring $M_n(\mathbb{R})$. Fix any matrix $\bM \in M_n(\mathbb{R})$, then $\mathbb{R}^n$ is a module over the polynomial ring $\mathbb{R}[t]$ with $P(t)\cdot \bx= P(\bM)\bx$ for any polynomial $P(t)$ and vector $\bx$. If the normal form of $\bM$ has a Jordan block of size $m$ with eigenvalue $\lambda$, then $\mathbb{R}^n$ has a $m$-dimenisonal submodule over $\mathbb{R}[t]$. In general, $\mathbb{R}^n$ is the direct sum of irreducible submodules, each corresponding to a Jordan block of the normal form of $\bM$.
    \item For a ring $R$, it is a module over itself with module scalar multiplication the same as the ring multiplication. A submodule $I\subset R$ is called an \emph{ideal} of $R$. For example, the even numbers $\{\ldots, -4, -2, 0, 2,4,\ldots\}$ is an ideal of the integer ring $\mathbb{Z}$. 
    
    If $R$ is commutative, an ideal $\mathfrak{p}$ is \emph{prime} if it is not the product of different ideals other than $0$ and $R$. For example, in $\mathbb{Z}$, a prime ideal takes the form $\{pm \mid m\in \mathbb{Z}\}$ for some prime number $p$. In commutative algebra and algebraic geometry, the set of prime ideals is called the \emph{spectrum} of $R$ \cite{Ati94}. 
    \end{enumerate}
\end{Example}

In each of \cref{eg:auc} \ref{it:tso})-\ref{it:tsoi}), there is a hidden base ring $R$. This gives an additional structure to either example, which we formalize as follows.

\begin{Definition}
    Let $R$ be a commutative ring (usually chosen as a field such as $\mathbb{R}$ or $\mathbb{C}$). A ring $A$ is an $R$-algebra or an algebra over $R$ if $A$ is a module over $R$ such that
    $(ra)a' = r(aa')$ for any $r\in R, a, a'\in A$. 
\end{Definition}

If we have the same algebraic structure such as ring, module, or algebra on sets $S_1, S_2$, then a function $\phi: S_1 \to S_2$ is a \emph{homomorphism} if it preserves the respective algebraic structure. It is an \emph{isomorphism} if it is invertible, which essentially identifies $S_1$ and $S_2$.   

Having introduced all the necessary concepts, we can describe ASP. It requires the data $(\calA, \calM, \rho)$, where $\calA$ is an unital algebra (over a base field), $\calM$ is a vector space, and $\rho: \calA \to \text{End}(\calM)$ is an algebra homomorphism of $\calA$ into the endomorphism ring of $\calM$ (cf.\ \cref{eg:auc}~\ref{it:tso})). The map $\rho$ makes $\calM$ an $\calA$-module with 
\begin{align*}
    a\cdot x = \rho(a)x, a\in \calA, x\in \calM.
\end{align*}
It is usually assumed that $\calM$ is finitely generated
as an $\calA$-module.
\begin{Example} \label{eg:ite}
 In this example, we explain how GSP fits into the ASP framework. Consider a graph $G$ of size $n$ and $\bS$ a chosen graph shift operator (GSO). Similar to \cref{eg:tfs}~\ref{it:tvs}), let $\calA$ be the polynomial ring $\mathbb{R}[t]$ in a single variable $t$ and $\calM$ be the signal space $\mathbb{R}^n$ on $G$. Then the endomorphism ring $\text{End}(\calM)$ is nothing but the matrix ring $M_n(\mathbb{R})$. The map $\rho: \mathbb{R}[t] \to M_n(\mathbb{R}), t\mapsto \bS$ is the setup for graph signal processing (GSP). It is worth pointing out that the ring $\calA$ is a principal ideal domain, whose structure theorem says that $\calM$ can be decomposed into $1$-dimenision irreducible submodules \cite{Art11}.   
\end{Example}

ASP has the following interpretations of the key concepts we are interested in:
\begin{itemize}
    \item In ASP, the \emph{Fourier transform} is a decomposition of $\calM$ into a direct sum of ($1$-dimensional) irreducible submodules $\Delta: \calM \to \oplus_{w\in W}\calM_w$, where $W$ is the index set that also corresponds to the coordinates of the \emph{frequency domain}. 
    \item In ASP, a \emph{convolution} is $\rho(a) \in \text{End}(\calM), a\in \calA$. This corresponds to the property that, in GSP, a convolution is a polynomial in the generators of $\calA$.
    \item In ASP, a bandlimited space is a submodule of $\calM$, which is isomorphic to the direct sum of irreducible submodules.
\end{itemize}

The framework places the algebra $\rho(\mathcal{A})$ at the center of the stage. As we have seen, in GSP, $\mathcal{A} = \mathbb{R}[t]$ and $\rho(t) = \bS$ for a chosen shift operator. In contrast, we illustrate how this idea can be applied for lattice signal processing \cite{Pus21} without using a single shift operator. 

\begin{Example}
We describe lattice signal processing in this example. Instead of giving all the details, we explain how the shifts and the associated matrix algebra are defined, which is sufficient to apply the ASP framework. 

Let $\calL$ be a \emph{meet semilattice} of size $n$. This means that $\calL$ has a partial order $\leq$ that ``compares'' elements of $\calL$, with the following properties
\begin{enumerate}
    \item $a\leq a$ for any $a\in \calL$.
    \item If $a\leq b, b\leq a$, then $a=b$.
    \item If $a\leq b, b\leq c$, then $a\leq c$. 
    \item For $a,b\in \calL$, there is a unique greatest lower bound (in terms of $\leq$) $a \wedge b$, called the \emph{meet}. 
\end{enumerate}
A signal $\bs = \{s_a,a\in \calL\}$ assigns a number $s_a$ to each $a\in \calL$. For each $a\in \calL$, we define a shift $\bT_a$ by requiring
\begin{align*}
    \bT_a(\bs) = (s_{b\wedge a})_{b\in \calL}. 
\end{align*}
Intuitively, the signal on each $b\in \calL$ is ``shifted'' to its meet with $a$. 

One verifies that $\{\bT_a,a\in\calL\}$ are pairwise commutative as $\bT_a\circ \bT_{a'} = \bT_{a\wedge a'}=\bT_{a'\wedge a}$. The shifts have a common eigenbasis. By taking compositions and linear combinations in $M_n(\mathbb{R})$, they form the required algebra $\calA$ (with $\rho$ the identity map). As an application, the theory can be used for sampling in auction \cite{Pus21}.
\end{Example}

\section{Towards a unification} \label{sec:au}

So far, we have briefly described two major approaches to having a ``Fourier theory'', namely, the Pontryagin theory and the theory of ASP. We summarize their major differences in the following table. 

\begin{table}[htbp]
\caption{The difference between the Pontryagin theory and the theory of ASP.} \label{tab:tdb}
\begin{center}
\scalebox{1.2}{
\begin{tabular}{ |c|c|c| } 
 \hline
   & Pontryagin & ASP \\ 
 \hline
 \hline
 Signal space & Infinite dimensional & Finite dimensional \\ 
 \hline
 Domain symmetry  & Group & Irregular \\ 
 \hline
 Fourier kernel & Characters & Vector space basis \\ 
 \hline 
 Measure & Haar & ``Not required'' \\
 \hline
 Tools & Algebra, topology, analysis & Algebra, combinatorics \\ 
 \hline
\end{tabular}}
\end{center}
\end{table}

Technically, the Pontryagin theory and ASP are loosely hinged on DFT, which can be interpreted with both a finite cyclic group and a finite directed cycle graph. The relations among them and their specializations can be visualized with a Venn diagram \figref{fig:relations}.   

\begin{figure}[h]
    \centering
    \includegraphics[width=0.5\textwidth]{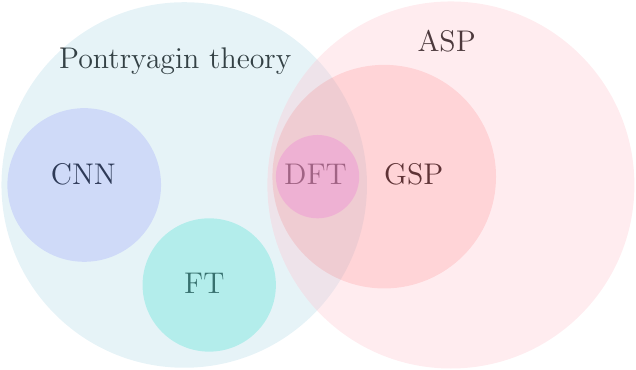}
    \caption{Relations among the theories discussed.}
    \label{fig:relations}
\end{figure}

However, intuitively, any Fourier theory intends to analyze a signal by inspecting its ``response'' to different frequency modes. Hence, there should be a certain point of view that we can (partially) unify all the theories technically. The classical Fourier theory starts with the Fourier transform. While convolution can be introduced independently, it is usually associated with Fourier transform as multiplication in the frequency domain. ASP turns the table around by starting with a module structure of a vector space $\calM$ over an algebra $\calA$. The frequency domain is subsequently introduced as the index set of a module decomposition. Inspired by this and a classical result of Gelfand, we propose to understand the theories in terms of the convolution algebra of the signal space as follows. 

For simplicity, we consider mainly the Pontryagin theory and GSP. If $A$ is a locally compact abelian group, we consider the signal space $\calS = L^1(A)$. On the other hand, for GSP, we consider the signal space $\calS = \mathbb{R}^n$. In the former case, $\calS$ is a Banach algebra under convolution as in \cref{coro:lia}. For GSP, we define an algebra structure $\calS = \mathbb{R}^n$, called the convolution algebra, with multiplication given by the convolution $*$ as in (\ref{eg:wzw}). 

A \emph{character} of $\calS$ is a continuous algebra homomorphism from $\calS$ to the base field, which is $\mathbb{C}$ and $\mathbb{R}$ for the Pontryagin theory and GSP respectively. The set of nonzero characters is denoted by $\widehat{\calS}$. In theory, $\widehat{\calS}$ is uniquely determined by $\calS$. Furthermore, $\calS$ determines the Fourier transform in the following sense. 

\begin{Theorem} \label{thm:tia}
    There is a bijection between $\widehat{\calS}$ and the Fourier kernel. 
\end{Theorem}

\begin{proof}
    If $A$ is a locally compact abelian group and $\chi \in \widehat{A}$, we define a character $\nu_{\chi}: \calS \to \mathbb{C}$ by the formula:
    \begin{align*}
        \nu_{\chi}(f) = \widehat{f}(\chi)=\int_A f(a)\overline{\chi}(a) \ud a. 
    \end{align*}
This defines a bijection between $\widehat{A}$ and $\widehat{S}$ according to \cite{Lax02} p.218 or \cite{Ram99} 3-11.

For GSP, recall the Fourier transform is given by a base change $\bU: \mathbb{R}^n \to \mathbb{R}^n$, where $\bU$ is an orthogonal matrix. The algebra structure on $\calS$ can be explicitly expressed as $\bU\big(\bU^{-1}(\bx)\odot\bU^{-1}(\by)\big)$. It suffices to show that the columns $\bu_i, 1\leq i\leq n$ of $\bU$ are uniquely determined by $\widehat{\calS}$. 

For each $\bu_i$, we associate it with $\nu_i \in \widehat{\calS}$ by the formula $\nu_i(\bx) = \langle \bu_i,\bx\rangle$. As it is just the $i$-th component of GFT, $\nu_i$ respects the mutliplication in $\calS$ and hence $\nu_i \in \widehat{\calS}$. 

In the reverse direction, for $\nu \in \widehat{\calS}$, it is a linear transformation $\calS \to \mathbb{R}$. Therefore, there is a nonzero vector $\bx_{\nu}$ such that $\nu(\by) = \langle \bx_{\nu}, \by\rangle$ for any $\by \in \calS$. Moreover, $\nu$ respects the multiplication of $\calS$, and hence for $\bu_i,\bu_j, 1\leq i,j\leq n$, we have 
\begin{align} \label{eq:lbb}
    \nu(\bu_i)\nu(\bu_j)=\langle \bx_{\nu},\bu_i\rangle\langle \bx_{\nu},\bu_j\rangle = \langle \bx_{\nu},\bu_i*\bu_j\rangle = \nu(\bu_i*\bu_j). 
\end{align}
Notice that $\bu_i*\bu_j=\bu_i$ if $i= j$ and $0$ otherwise. Therefore, for some $1\leq i_{\nu}\leq n$, $\langle \bx_{\nu},\bu_{i_{\nu}}\rangle \neq 0$ and $\langle \bx_{\nu},\bu_{i_{\nu}}\rangle^2= \langle \bx_{\nu},\bu_{i_{\nu}}\rangle$. We have $\langle \bx_{\nu},\bu_{i_{\nu}}\rangle = 1$. Moreover, for $j \neq i_{\nu}$, (\ref{eq:lbb}) implies that 
\begin{align*}
   \langle \bx_{\nu},\bu_j\rangle= \langle \bx_{\nu},\bu_{i_{\nu}}\rangle\langle \bx_{\nu},\bu_j\rangle = \langle \bx_{\nu},\bu_{i_{\nu}}*\bu_j\rangle = 0. 
\end{align*}
Hence, $\langle \bx_{\nu},\bu_j\rangle=0$ for any $j\neq i_{\nu}$. Consequently, $\x_{\nu}= \bu_{i_{\nu}}$. This implies that $\nu$ determines a unique column of $\bU$ as claimed.  
\end{proof}

This observation demonstrates that convolution can indeed be unifyingly viewed as an indispensable component of Fourier theories with both finite and infinite domains. We end the paper with an example to illustrate this point of view.

\begin{Example}
For simplicity, let $T$ be a finite (parameter) set of size $m$ and $V$ be a set of $n$ nodes. As both $T$ and $V$ are finite, the space of (real) integrable functions on $T\times V$ can be identified with $\calS = \mathbb{R}^{mn}$. Each signal in $\calS$ can be written as an $n\times m$ matrix, with the $t$-th column corresponding to the parameter $t\in T$ assuming $T$ is given an order.

Assume that for each $t\in T$, there is a symmetric graph shift operator $\bS_t$ with an orthogonal matrix of eigenvectors $\bU_t$. This can happen if there exist multiple feasible connections among $V$, e.g., multi-layered graphs and heterogeneous graphs. We define a multiplication on $\calS$ as follows. For two matrices $\bM,\bN\in \calS$, the $t$-th column of $\bM*\bN$ is the GSP convolution of the $t$-th columns of $\bM,\bN$ w.r.t.\ $\bU_t$. As $\calS$ can be vectorized, by the same argument as in the proof of \cref{thm:tia}, one shows that $\widehat{\calS}$ can be identified with $\mathfrak{B} = \{\bU_{t,j}, t\in T, 1\leq j\leq n\}$, where $\bU_{t,j}$ is the $j$-th column of $\bU_t$. The explicit formula for the character is given by
\begin{align*}
    \nu_{t,j}(\bM) = \langle \bU_{t,j}, \bM_t\rangle,
\end{align*}
where $\bM_t$ is the $j$-th column of $\bM$. In the spirit of the unification, $\mathfrak{B}$ can be considered as the Fourier kernel for a transformation $\widehat{\cdot}: \calS \to \calS$. This agrees with what has been proposed in \cite{Ji23}, while the latter considers more general parameter spaces $T$ other than finite sets.  

In addition, in \cite{Ji23}, $T$ is assumed to be a probability space. To process a signal on $V$, we only need to introduce two operators $\phi_1: \mathbb{R}^n \to \calS$ and $\phi_2: \calS \to \mathbb{R}^n$ as ``dictionaries'' to translate between the two signal spaces. For example, we may choose $\phi_1(\bx)$ such that the vector associated with each $t\in T$ is $\bx$, while $\phi_2$ is the expectation operator $\mathbb{E}(\cdot)$ w.r.t.\ the measure on $T$. We may use $\phi = \phi_2\circ \widehat{\phi_1(\cdot)}: \mathbb{R}^n \to \mathbb{R}^n$ to analyze signals analogous to GFT in GSP, however, it is in general not an orthogonal base change. 
\end{Example}

\section{Conclusion}
In this article, we give an overview of the notion of convolution in theories of Fourier type such as the Pontryagin theory, the theories of GSP and ASP. It appears in different forms. A unifying view is provided. 

\bibliographystyle{IEEEtran}
\bibliography{allref}

\end{document}

%% file: preamble.tex
\usepackage[T1]{fontenc}
\usepackage[utf8]{inputenc}
\usepackage{mathtools}

\usepackage{amssymb,mathrsfs}
\usepackage{amsthm}
\usepackage{bm}
\usepackage{scalerel}
\usepackage{nicefrac}
\usepackage{microtype} 
\usepackage[shortlabels]{enumitem}
\usepackage{graphicx}
\usepackage{epstopdf}
\DeclareGraphicsExtensions{.eps,.png,.jpg,.pdf}

\usepackage{url}
\usepackage{colortbl}
\usepackage{booktabs}
\usepackage{multirow}
\usepackage{colortbl,xcolor}
\usepackage[normalem]{ulem}
\usepackage{xparse,xstring}
\usepackage{calc}
\usepackage{etoolbox}

\makeatletter
\@ifpackageloaded{natbib}{
	\relax
}{
	\usepackage{cite}
}
\makeatother

\usepackage{array}
\newcolumntype{L}[1]{>{\raggedright\let\newline\\\arraybackslash\hspace{0pt}}m{#1}}
\newcolumntype{C}[1]{>{\centering\let\newline\\\arraybackslash\hspace{0pt}}m{#1}}
\newcolumntype{R}[1]{>{\raggedleft\let\newline\\\arraybackslash\hspace{0pt}}m{#1}}

\makeatletter
\let\MYcaption\@makecaption
\makeatother
\usepackage[font=footnotesize]{subcaption}
\makeatletter
\let\@makecaption\MYcaption
\makeatother

\usepackage{glossaries}
\makeatletter
\sfcode`\.1006

\let\oldgls\gls
\let\oldglspl\glspl

\newcommand\fussy@ifnextchar[3]{%
	\let\reserved@d=#1%
	\def\reserved@a{#2}%
	\def\reserved@b{#3}%
	\futurelet\@let@token\fussy@ifnch}
\def\fussy@ifnch{%
	\ifx\@let@token\reserved@d
		\let\reserved@c\reserved@a
	\else
		\let\reserved@c\reserved@b
	\fi
	\reserved@c}

\renewcommand{\gls}[1]{%
\oldgls{#1}\fussy@ifnextchar.{\@checkperiod}{\@}}
\renewcommand{\glspl}[1]{%
\oldglspl{#1}\fussy@ifnextchar.{\@checkperiod}{\@}}

\newcommand{\@checkperiod}[1]{%
	\ifnum\sfcode`\.=\spacefactor\else#1\fi
}

\robustify{\gls}
\robustify{\glspl}
\makeatother

\newacronym{wrt}{w.r.t.}{with respect to}
\newacronym{RHS}{R.H.S.}{right-hand side}
\newacronym{LHS}{L.H.S.}{left-hand side}
\newacronym{iid}{i.i.d.}{independent and identically distributed}
\newacronym{SOTA}{SOTA}{state-of-the-art}

\usepackage{float}

\ifx\notloadhyperref\undefined
	\ifx\loadbibentry\undefined
		\usepackage[hidelinks,hypertexnames=false]{hyperref}
	\else
		\usepackage{bibentry}
		\makeatletter\let\saved@bibitem\@bibitem\makeatother
		\usepackage[hidelinks,hypertexnames=false]{hyperref}
		\makeatletter\let\@bibitem\saved@bibitem\makeatother
	\fi
\else
	\ifx\loadbibentry\undefined
		\relax
	\else
		\usepackage{bibentry}
	\fi
\fi

\usepackage[capitalize]{cleveref}
\crefname{equation}{}{}
\Crefname{equation}{}{}
\crefname{claim}{claim}{claims}
\crefname{step}{step}{steps}
\crefname{line}{line}{lines}
\crefname{condition}{condition}{conditions}
\crefname{dmath}{}{}
\crefname{dseries}{}{}
\crefname{dgroup}{}{}

\crefname{Problem}{Problem}{Problems}
\crefformat{Problem}{Problem~#2#1#3}
\crefrangeformat{Problem}{Problems~#3#1#4 to~#5#2#6}

\crefname{Theorem}{Theorem}{Theorems}
\crefname{Corollary}{Corollary}{Corollaries}
\crefname{Proposition}{Proposition}{Propositions}
\crefname{Lemma}{Lemma}{Lemmas}
\crefname{Definition}{Definition}{Definitions}
\crefname{Example}{Example}{Examples}
\crefname{Assumption}{Assumption}{Assumptions}
\crefname{Remark}{Remark}{Remarks}
\crefname{Rem}{Remark}{Remarks}
\crefname{remarks}{Remarks}{Remarks}
\crefname{Appendix}{Appendix}{Appendices}
\crefname{Supplement}{Supplement}{Supplements}
\crefname{Exercise}{Exercise}{Exercises}
\crefname{Theorem_A}{Theorem}{Theorems}
\crefname{Corollary_A}{Corollary}{Corollaries}
\crefname{Proposition_A}{Proposition}{Propositions}
\crefname{Lemma_A}{Lemma}{Lemmas}
\crefname{Definition_A}{Definition}{Definitions}

\usepackage{crossreftools}
\ifx\notloadhyperref\undefined
\else
	\relax
\fi

\usepackage{algorithm}
\usepackage{algpseudocode}

\ifx\loadbreqn\undefined
	\relax
\else
	\usepackage{breqn}
\fi

\makeatletter
\def\cleartheorem#1{%
    \expandafter\let\csname#1\endcsname\relax
    \expandafter\let\csname c@#1\endcsname\relax
}
\def\clearthms#1{ \@for\tname:=#1\do{\cleartheorem\tname} }
\makeatother

\ifx\renewtheorem\undefined
	\ifx\useTheoremCounter\undefined
		\newtheorem{Theorem}{Theorem}
		\newtheorem{Corollary}{Corollary}
		\newtheorem{Proposition}{Proposition}
		
	\else
		\newtheorem{Theorem}{Theorem}
		\newtheorem{Corollary}[Theorem]{Corollary}
		
	\fi

	\newtheorem{Definition}{Definition}
	\newtheorem{Example}{Example}

\fi

\theoremstyle{remark}

\theoremstyle{plain}

\newcommand{\qednew}{\nobreak \ifvmode \relax \else
		\ifdim\lastskip<1.5em \hskip-\lastskip
			\hskip1.5em plus0em minus0.5em \fi \nobreak
		\vrule height0.75em width0.5em depth0.25em\fi}



\NewDocumentCommand{\movedownsub}{e{^_}}{%
	\IfNoValueTF{#1}{%
		\IfNoValueF{#2}{^{}}
	}{%
		^{#1}
	}%
	\IfNoValueF{#2}{_{#2}}
}

\let\latexchi\chi
\RenewDocumentCommand{\chi}{}{\latexchi\movedownsub}

\newcommand{\iu}{\mathfrak{i}\mkern1mu}

\newcommand{\calA}{\mathcal{A}}

\newcommand{\calL}{\mathcal{L}}
\newcommand{\calM}{\mathcal{M}}

\newcommand{\calS}{\mathcal{S}}

\newcommand{\bA}{\mathbf{A}}

\newcommand{\bM}{\mathbf{M}}

\newcommand{\bN}{\mathbf{N}}

\newcommand{\bs}{\mathbf{s}}
\newcommand{\bS}{\mathbf{S}}

\newcommand{\bT}{\mathbf{T}}
\newcommand{\bu}{\mathbf{u}}
\newcommand{\bU}{\mathbf{U}}

\newcommand{\bW}{\mathbf{W}}
\newcommand{\bx}{\mathbf{x}}
\newcommand{\bX}{\mathbf{X}}
\newcommand{\by}{\mathbf{y}}

\newcommand{\bz}{\mathbf{z}}

\DeclareSymbolFont{bsfletters}{OT1}{cmss}{bx}{n}
\DeclareSymbolFont{ssfletters}{OT1}{cmss}{m}{n}
\DeclareMathSymbol{\bsfGamma}{0}{bsfletters}{'000}
\DeclareMathSymbol{\ssfGamma}{0}{ssfletters}{'000}
\DeclareMathSymbol{\bsfDelta}{0}{bsfletters}{'001}
\DeclareMathSymbol{\ssfDelta}{0}{ssfletters}{'001}
\DeclareMathSymbol{\bsfTheta}{0}{bsfletters}{'002}
\DeclareMathSymbol{\ssfTheta}{0}{ssfletters}{'002}
\DeclareMathSymbol{\bsfLambda}{0}{bsfletters}{'003}
\DeclareMathSymbol{\ssfLambda}{0}{ssfletters}{'003}
\DeclareMathSymbol{\bsfXi}{0}{bsfletters}{'004}
\DeclareMathSymbol{\ssfXi}{0}{ssfletters}{'004}
\DeclareMathSymbol{\bsfPi}{0}{bsfletters}{'005}
\DeclareMathSymbol{\ssfPi}{0}{ssfletters}{'005}
\DeclareMathSymbol{\bsfSigma}{0}{bsfletters}{'006}
\DeclareMathSymbol{\ssfSigma}{0}{ssfletters}{'006}
\DeclareMathSymbol{\bsfUpsilon}{0}{bsfletters}{'007}
\DeclareMathSymbol{\ssfUpsilon}{0}{ssfletters}{'007}
\DeclareMathSymbol{\bsfPhi}{0}{bsfletters}{'010}
\DeclareMathSymbol{\ssfPhi}{0}{ssfletters}{'010}
\DeclareMathSymbol{\bsfPsi}{0}{bsfletters}{'011}
\DeclareMathSymbol{\ssfPsi}{0}{ssfletters}{'011}
\DeclareMathSymbol{\bsfOmega}{0}{bsfletters}{'012}
\DeclareMathSymbol{\ssfOmega}{0}{ssfletters}{'012}

\newcommand{\bLambda}{\bm{\Lambda}}

\makeatletter
\newcommand*\rel@kern[1]{\kern#1\dimexpr\macc@kerna}
\newcommand*\widebar[1]{%
  \begingroup
  \def\mathaccent##1##2{%
    \rel@kern{0.8}%
    \overline{\rel@kern{-0.8}\macc@nucleus\rel@kern{0.2}}%
    \rel@kern{-0.2}%
  }%
  \macc@depth\@ne
  \let\math@bgroup\@empty \let\math@egroup\macc@set@skewchar
  \mathsurround\z@ \frozen@everymath{\mathgroup\macc@group\relax}%
  \macc@set@skewchar\relax
  \let\mathaccentV\macc@nested@a
  \macc@nested@a\relax111{#1}%
  \endgroup
}
\makeatother

\DeclareMathOperator{\var}{var}

\DeclareMathOperator{\cov}{cov}

\newcommand{\ifbcdot}[1]{\ifblank{#1}{\cdot}{#1}}

\DeclarePairedDelimiterX\abs[1]{\lvert}{\rvert}{\ifbcdot{#1}}
\DeclarePairedDelimiterX\parens[1]{(}{)}{\ifbcdot{#1}}
\DeclarePairedDelimiterX\brk[1]{[}{]}{\ifbcdot{#1}}
\DeclarePairedDelimiterX\braces[1]{\{}{\}}{\ifbcdot{#1}}
\DeclarePairedDelimiterX\angles[1]{\langle}{\rangle}{\ifblank{#1}{\cdot,\cdot}{#1}}
\DeclarePairedDelimiterX\ip[2]{\langle}{\rangle}{\ifbcdot{#1},\ifbcdot{#2}}
\DeclarePairedDelimiterX\norm[1]{\lVert}{\rVert}{\ifbcdot{#1}}
\DeclarePairedDelimiterX\ceil[1]{\lceil}{\rceil}{\ifbcdot{#1}}
\DeclarePairedDelimiterX\floor[1]{\lfloor}{\rfloor}{\ifbcdot{#1}}

\DeclareFontFamily{U}{matha}{\hyphenchar\font45}
\DeclareFontShape{U}{matha}{m}{n}{
      <5> <6> <7> <8> <9> <10> gen * matha
      <10.95> matha10 <12> <14.4> <17.28> <20.74> <24.88> matha12
      }{}
\DeclareSymbolFont{matha}{U}{matha}{m}{n}
\DeclareFontSubstitution{U}{matha}{m}{n}

\DeclareFontFamily{U}{mathx}{\hyphenchar\font45}
\DeclareFontShape{U}{mathx}{m}{n}{
      <5> <6> <7> <8> <9> <10>
      <10.95> <12> <14.4> <17.28> <20.74> <24.88>
      mathx10
      }{}
\DeclareSymbolFont{mathx}{U}{mathx}{m}{n}
\DeclareFontSubstitution{U}{mathx}{m}{n}

\DeclareMathDelimiter{\vvvert}{0}{matha}{"7E}{mathx}{"17}
\DeclarePairedDelimiterX\vertiii[1]{\vvvert}{\vvvert}{\ifbcdot{#1}}

\DeclarePairedDelimiterXPP\trace[1]{\operatorname{Tr}}{(}{)}{}{\ifbcdot{#1}} 
\DeclarePairedDelimiterXPP\col[1]{\operatorname{col}}{\{}{\}}{}{\ifbcdot{#1}} 
\DeclarePairedDelimiterXPP\row[1]{\operatorname{row}}{\{}{\}}{}{\ifbcdot{#1}} 
\DeclarePairedDelimiterXPP\erf[1]{\operatorname{erf}}{(}{)}{}{\ifbcdot{#1}}
\DeclarePairedDelimiterXPP\erfc[1]{\operatorname{erfc}}{(}{)}{}{\ifbcdot{#1}}
\DeclarePairedDelimiterXPP\KLD[2]{D}{(}{)}{}{\ifbcdot{#1}\, \delimsize\|\, \ifbcdot{#2}} 
\DeclarePairedDelimiterXPP\op[2]{\operatorname{#1}}{(}{)}{}{#2} 

\newcommand{\ud}{\,\mathrm{d}} 

\DeclarePairedDelimiterXPP\indicate[1]{{\bf 1}}{\{}{\}}{}{\ifbcdot{#1}}

\NewDocumentCommand\ofrac{s m}{%
	\IfBooleanTF#1%
	{\dfrac{1}{#2}}%
	{\frac{1}{#2}}%
}
\NewDocumentCommand\ddfrac{s m m}{%
	\IfBooleanTF#1%
	{\dfrac{\mathrm{d} {#2}}{\mathrm{d} {#3}}}%
	{\frac{\mathrm{d} {#2}}{\mathrm{d} {#3}}}%
}
\NewDocumentCommand\ppfrac{s m m}{%
	\IfBooleanTF#1%
	{\dfrac{\partial {#2}}{\partial {#3}}}%
	{\frac{\partial {#2}}{\partial {#3}}}%
}

\providecommand\given{}

\DeclarePairedDelimiterX\Set[2]\{\}{%
\renewcommand\given{\SetSymbol[\delimsize]{#1}}
#2
}
\DeclarePairedDelimiterX\Setc[1]\{\}{%
\renewcommand\given{\SetSymbol{:}}
#1
}

\NewDocumentCommand\set{s o m}{%
	\IfBooleanTF#1%
	{\IfValueTF{#2}{\Set*{#2}{#3}}{\Setc*{#3}}}%
	{\IfValueTF{#2}{\Set{#2}{#3}}{\Setc{#3}}}%
}

\NewDocumentCommand{\evalat}{ s O{\big} m e{_^} }{%
\IfBooleanTF{#1}%
{\left. #3 \right|}{#3#2|}%
\IfValueT{#4}{_{#4}}%
\IfValueT{#5}{^{#5}}%
}

\providecommand\given{}
\DeclarePairedDelimiterXPP\cprob[1]{}(){}{
\renewcommand\given{\nonscript\,\delimsize\vert\allowbreak\nonscript\,\mathopen{}}%
#1%
}
\DeclarePairedDelimiterXPP\cexp[1]{}[]{}{
\renewcommand\given{\nonscript\,\delimsize\vert\allowbreak\nonscript\,\mathopen{}}%
#1%
}

\DeclareDocumentCommand \P { s e{_^} d() g } {%
	\mathbb{P}%
	\IfBooleanTF{#1}%
		{
			\IfValueT{#2}{_{#2}}%
			\IfValueT{#3}{^{#3}}%
			\IfValueTF{#5}{\cprob{#4 \given #5}}{\IfValueT{#4}{\cprob{#4}}}%
		}%
		{
			\IfValueT{#2}{_{#2}}%
			\IfValueT{#3}{^{#3}}%
			\IfValueTF{#5}{\cprob*{#4 \given #5}}{\IfValueT{#4}{\cprob*{#4}}}%
		}%
}

\DeclareDocumentCommand \E { s e{_^} o g } {%
	\mathbb{E}%
	\IfBooleanTF{#1}%
		{
			\IfValueT{#2}{_{#2}}%
			\IfValueT{#3}{^{#3}}%
			\IfValueTF{#5}{\cexp{#4 \given #5}}{\IfValueT{#4}{\cexp{#4}}}%
		}%
		{
			\IfValueT{#2}{_{#2}}%
			\IfValueT{#3}{^{#3}}%
			\IfValueTF{#5}{\cexp*{#4 \given #5}}{\IfValueT{#4}{\cexp*{#4}}}%
		}%
}

\DeclareDocumentCommand \Var { s e{_^} d() g } {%
	\var%
	\IfBooleanTF{#1}%
		{
			\IfValueT{#2}{_{#2}}%
			\IfValueT{#3}{^{#3}}%
			\IfValueTF{#5}{\cprob{#4 \given #5}}{\IfValueT{#4}{\cprob{#4}}}%
		}%
		{
			\IfValueT{#2}{_{#2}}%
			\IfValueT{#3}{^{#3}}%
			\IfValueTF{#5}{\cprob*{#4 \given #5}}{\IfValueT{#4}{\cprob*{#4}}}%
		}%
}

\DeclareDocumentCommand \Cov { s e{_^} d() g } {%
	\cov%
	\IfBooleanTF{#1}%
		{
			\IfValueT{#2}{_{#2}}%
			\IfValueT{#3}{^{#3}}%
			\IfValueTF{#5}{\cprob{#4 \given #5}}{\IfValueT{#4}{\cprob{#4}}}%
		}%
		{
			\IfValueT{#2}{_{#2}}%
			\IfValueT{#3}{^{#3}}%
			\IfValueTF{#5}{\cprob*{#4 \given #5}}{\IfValueT{#4}{\cprob*{#4}}}%
		}%
}

\ExplSyntaxOn
\NewDocumentCommand \dist {m o o} {%
\mathrm{#1}\left(%
	\IfValueT{#3}{%
		\tl_if_blank:nTF{ #3 }{\cdot\, \middle|\, }{#3\, \middle|\, }%
	}
	\IfValueT{#2}{#2}%
\right)%
}
\ExplSyntaxOff

\NewDocumentCommand {\cbrace} {t+ D[]{black} D(){\widthof{#5}} m m } {%
	\begingroup%
		\color{#2}
		\IfBooleanTF{#1}{%
			\overbrace{#4}^%
		}{
			\underbrace{#4}_%
		}%
		{\parbox[c]{#3}{\centering\footnotesize{#5}}}%
	\endgroup%
}

\let\oldforall\forall
\renewcommand{\forall}{\oldforall \, }

\let\oldexist\exists
\renewcommand{\exists}{\oldexist \, }

\makeatletter

\newcommand{\rankcolor}[2]{%
	\expandafter\renewcommand\csname #1\endcsname[1]{%
		\ifblank{##1}{%
			{\color{#2} \textbf{#2}}%
		}{%
			\ifmmode
				\textcolor{#2}{\bm{##1}}%
			\else%
				{\color{#2} \textbf{##1}}%
			\fi	
		}%
	}
}

\rankcolor{first}{red}
\rankcolor{second}{blue}
\rankcolor{third}{cyan}
\makeatother

\newcommand{\figref}[1]{Fig.~\ref{#1}}
\graphicspath{{./Figures/}{./figures/}}
\DeclareDocumentCommand{\includeCroppedPdf}{ o O{./Figures/} m }{
	\IfFileExists{#2#3-crop.pdf}{}{%
		\immediate\write18{pdfcrop #2#3.pdf #2#3-crop.pdf}}%
	\includegraphics[#1]{#2#3-crop.pdf}
}

\makeatletter
\newcommand*{\addFileDependency}[1]{
  \typeout{(#1)}
  \@addtofilelist{#1}
  \IfFileExists{#1}{}{\typeout{No file #1.}}
}
\makeatother

\definecolor{gray90}{gray}{0.9}
\def\colorlist{red,blue,brown,cyan,darkgray,gray,lightgray,green,lime,magenta,olive,orange,pink,purple,teal,violet,white,yellow}

\makeatletter
\def\startcomment{[}
\ifx\nohighlights\undefined
	\newcommand{\createcolor}[1]{%
			\expandafter\newcommand\csname #1\endcsname[1]{{\color{#1} ##1}}%
	}
	\newcommand{\msout}[1]{\text{\color{green} \sout{\ensuremath{#1}}}}
	\newcommand{\del}[1]{{\color{green}\ifmmode \msout{#1}\else\sout{#1}\fi}}
\else
	\newcommand{\createcolor}[1]{%
			\expandafter\newcommand\csname #1\endcsname[1]{%
				\noexpandarg%
				\StrChar{##1}{1}[\firstletter]%
				\if\firstletter\startcomment%
					\relax
				\else%
					##1
				\fi
			}%
	}
	\newcommand{\msout}[1]{}
	\newcommand{\del}[1]{}
\fi

\def\@tempa#1,{%
    \ifx\relax#1\relax\else
        \createcolor{#1}%
        \expandafter\@tempa
    \fi
}
\expandafter\@tempa\colorlist,\relax,
\makeatother

\newcommand{\hhide}[1]{}

\ifx\diagnoselabel\undefined
	\relax
\else
	\makeatletter
	\def\@testdef #1#2#3{%
		\def\reserved@a{#3}\expandafter \ifx \csname #1@#2\endcsname
			\reserved@a  \else
			\typeout{^^Jlabel #2 changed:^^J%
				\meaning\reserved@a^^J%
				\expandafter\meaning\csname #1@#2\endcsname^^J}%
			\@tempswatrue \fi}
	\makeatother
\fi
